\documentclass[11pt]{article}%
\usepackage{algorithm}
\usepackage{algpseudocode}
\usepackage{amssymb}
\usepackage{amsfonts}
\usepackage{amsmath}
\usepackage{graphicx}%
\providecommand{\U}[1]{\protect\rule{.1in}{.1in}}
\usepackage{amsthm}
 
\renewcommand\and{\end{tabular}\kern-\tabcolsep\ and\ \kern-\tabcolsep\begin{tabular}[t]{c}}
\let\origthanks\thanks
\renewcommand\thanks[1]{\begingroup\let\rlap\relax\origthanks{#1}\endgroup}

\newtheorem{theorem}{Theorem}

\usepackage{amsmath}

\newtheorem{definition}[theorem]{Definition}

\newtheorem{lemma}[theorem]{Lemma}

\begin{document}

\title{Maximum-Likelihood Network Reconstruction for SIS Processes is NP-Hard}
\author{Bastian Prasse\thanks{Faculty of Electrical Engineering, Mathematics and
Computer Science, P.O Box 5031, 2600 GA Delft, The Netherlands; \emph{email}:
b.prasse@tudelft.nl, p.f.a.vanmieghem@tudelft.nl} \and Piet Van Mieghem\footnotemark[1]}
\date{Delft University of Technology \\
July 23, 2018}
\maketitle
\begin{abstract}
The knowledge of the network topology is imperative to precisely describing the viral dynamics of an SIS epidemic process. In scenarios for which the network topology is unknown, one resorts to reconstructing the network from observing the viral state trace. This work focusses on the impact of the viral state observations on the computational complexity of the resulting network reconstruction problem. We propose a novel method of \emph{constructing} a specific class of viral state traces from which the inference of the presence or absence of links is either easy or difficult. In particular, we use this construction to prove that the maximum-likelihood SIS network reconstruction is NP-hard. The NP-hardness holds for any adjacency matrix of a graph which is connected.
\end{abstract}

\section{Introduction}
We consider the network reconstruction of the sampled-time susceptible-infected-susceptible (SIS) process in a maximum-likelihood (ML) sense as introduced in \cite{prasseSIS}. We assume that the infection rate $\beta$ and the curing rate $\delta$ are known and that no self-infections occur; hence, the self-infection rate is $\epsilon=0$. We denote the number of nodes by $N$ and the $N \times 1$ viral state vector at discrete time $k$ by $x[k]$. At any time $k$, a node $i$ is either infected or susceptible, which is denoted by $x_i[k]= 1$ and $x_i[k] = 0$, respectively. We confine ourselves to connected graphs and denote by $\mathcal{A}$ the set of all $N\times N$ symmetric adjacency matrices $A$ with the elements $a_{ij}$. These adjacency matrices $A \in\mathcal{A}$ correspond to undirected, unweighted and connected graphs without self-loops. 
 
The network reconstruction problem for sampled-time SIS process is stated in the ML sense \cite{prasseSIS}. In contrast to the \textit{true} adjacency matrix $A$, which generated the viral states $x[k]$, the \textit{optimisation variable} in the ML estimation problem is denoted as $\hat{A}$. The \textit{solution} to the ML estimation problem, i.e. the adjacency matrix $\hat{A}$ which maximises the likelihood, is denoted by $A_\textup{\textrm{ML}}$.

\begin{definition}[SIS Network Reconstruction]\label{def:netreconstruction}
Given the viral state observations $x[k] \in \{0, 1\}^N$ from time $k=1$ to $k=n$ which originate from a sampled-time SIS process on an unknown adjacency matrix $A \in \mathcal{A}$, find the adjacency matrix $A_\textup{\textrm{ML}}$ which maximises the log-likelihood:
 \begin{align} \label{maxlikelihood}
 \begin{aligned}
A_\textup{\textrm{ML}} = &\underset{\hat{A}}{\textup{ arg max }} & & \log \left(  \mathrm{Pr} \left[x[1], ..., x[n] \Big| \hat{A} \right] \right)& \\
 &\textup{ s.t.} & &  \hat{a}_{ij} \in \{0, 1\}, \quad i,j = 1, ..., N&\\
  && &  \hat{a}_{ij}= \hat{a}_{ji}  , \quad  i,j = 1, ..., N&\\
   & & &  \hat{a}_{ii} =0, \quad   i =1, ..., N&
\end{aligned}
 \end{align}
\end{definition}
An instance of the optimisation problem (\ref{maxlikelihood}) is fully specified by the viral state observations $x[k] \in \{0, 1\}^N$ from time $k=1$ to $k=n$, where usually the \textit{observation length} $n$ satisfies $n >> N$. 

To stress the dependency of the ML estimate $A_\textup{\textrm{ML}}$ on a given viral state sequence $x[1], ..., x[n]$, we may also denote the ML estimate by $A_\textup{\textrm{ML}}\left( x[1], ..., x[n] \right)$. The SIS network reconstruction (\ref{maxlikelihood}) gives rise to two fundamental problems:
\begin{enumerate}
\item  \textit{How many observations $n$ are required such that the ML estimate $A_\textup{\textrm{ML}}\left( x[1], ..., x[n] \right)$ achieves a given accuracy $\varepsilon >0$ with high probability $p_{\textup{\textrm{acc}}} \approx 1$?} 

\item \textit{How to design an algorithm that computes the ML estimate $A_\textup{\textrm{ML}}\left( x[1], ..., x[n] \right)$ for a given viral state sequence $x[1], ..., x[n]$? What is the computational complexity of the SIS network reconstruction (\ref{maxlikelihood})?}
\end{enumerate}

The first problem translates to finding the minimal observation length $n_\textrm{min}$ such that
\begin{align*} 
\mathrm{Pr}\left(\lVert A_\textup{\textrm{ML}}\left( x[1], ..., x[n] \right) - A \rVert  \le \varepsilon\right) \ge p_\textrm{acc} \quad \forall n \ge n_\textrm{min},
\end{align*}
where $\lVert \cdot \rVert$ denotes some matrix norm. By proposing a heuristic to solve the ML estimation (\ref{maxlikelihood}), the results in \cite{prasseSIS} indicate that the minimum observation length $n_\textrm{min}$ increases subexponentially with respect to the number of nodes $N$: $\log_{10}(n_\textrm{min}) \approx N^\alpha +b$ for some constants $\alpha$ and $b$.

The focus of this work is on the second question. We prove that the ML estimation (\ref{maxlikelihood}) is NP-hard with respect to the number of nodes $N$ for \textit{any} connected adjacency matrix $A \in \mathcal{A}$. The idea of the proof is as follows: We aim to show that there is a polynomial-time reduction from the maximum cut problem to the ML estimation for the sampled-time SIS process (\ref{maxlikelihood}). Since the maximum cut problem is NP-complete \cite{bodlaender1991complexity}, this polynomial-time reduction proves that the ML estimation (\ref{maxlikelihood}) is NP-hard. As introduced in Section \ref{sec:max_cut}, the maximum cut problem can be stated as zero-one unconstrained quadratic programme (UQP). By comparison, we make the observation that the zero-one UQP which results from the maximum cut problem resembles the ML estimation (\ref{maxlikelihood}). We show that for every graph $G$ of the maximum cut problem, there is an SIS viral state sequence $x[1], ..., x[n]$ such that solving the ML estimation (\ref{maxlikelihood}) is equivalent to solving the maximum cut problem on the graph $G$. The polynomial-time reduction is presented in Section \ref{sec:reductionmaxcut}.

\section{Sampled-Time SIS Process}
\label{sec:sampledtimeSIS}

We give a brief summary of the sampled-time SIS process, and we refer to \cite{prasseSIS} for a more detailed description. The sampled-time Markov chain with sampling time $T$ is a discrete-time Markov chain \cite{van2014performance}. The probabilities of the viral state transitions depend on the adjacency matrix $A$. There are three kinds of transitions possible in the sampled-time Markov chain of the SIS process. These transitions are listed below and their probabilities are inferred from the continuous-time SIS equations.
\begin{description}
\item[Curing of a node] A single node $i$ changes from the infected state at discrete time $k$ to the susceptible state at discrete time $k+1$. The probability of this transition is 
\begin{equation}\label{healTransition}
\mathrm{Pr}\Big[x_i[k+1] = 0 \Big| x_i [k] = 1, x[k], A \Big] = \delta_T , 
\end{equation}
where the curing probability $\delta_T$ equals $\delta T$. 
\item[Infection of a node] A single node $i$ changes from the susceptible state at time instant $k$ to the infected state at time instant $k+1$ with the probability 
\begin{align}
\mathrm{Pr}\Big[x_i[k+1] = 1 \Big|  x_i [k] = 0, x[k], A\Big] &= \beta_T N_i(A, k), \label{infectionTrans}
\end{align}
where $N_i(A, k)$ is the number of infected nodes adjacent to node $i$ in $A$ at time $k$ and the infection probability $\beta_T$ equals $\beta T$. The  number of infected nodes adjacent to node $i$ equals 
\begin{align*}
N_i(A, k) = \sum_{j = 1}^N a_{ij} x_j[k].
\end{align*}

\item[No Change] No node changes its viral state from time $k$ to time $k+1$. This constant transition occurs when neither a curing nor an infection takes place, and hence
\begin{align}
\mathrm{Pr}\left[ x[k+1] =  x [k] \Big| x [k], A \right]  &=  1- \mathrm{Pr}\left[ \text{A node cures at } k+1\Big| x[k],  A \right] \nonumber\\
&\quad- \mathrm{Pr}\left[ \text{A node gets infected at } k+1\Big| x[k],  A \right], \label{constTrans} 
\end{align}
where the probabilities on the right-hand side can be derived from (\ref{healTransition}) and (\ref{infectionTrans}).
\end{description}
To ensure that (\ref{healTransition}), (\ref{infectionTrans}) and (\ref{constTrans}) are feasible expressions for probabilities, they have to be in $[0, 1]$ for all adjacency matrices $A \in \mathcal{A}$ and for all viral states $x[k]$. In \cite{prasseSIS}, an upper bound on the sampling time $T$ was derived, such that (\ref{healTransition}), (\ref{infectionTrans}) and (\ref{constTrans}) are in $[0, 1]$, and we assume that the sampling time $T$ does not exceed this upper bound.

\section{Maximum Cut}
\label{sec:max_cut}
We consider an undirected and unweighted graph $G = (\mathcal{N}, \mathcal{L})$, where $\mathcal{N} = \{1, ..., N\}$ is the set of nodes and $\mathcal{L}$ is the set of $L$ links. A cut-set of the graph $G$ is defined as follows \cite{karelCutSizeBounds, van2015epidemic}.
\begin{definition}[Cut-set]\label{def:cutset}
For a non-empty node subset $\mathcal{V} \subset \mathcal{N}$ of a graph and its complement $\bar{\mathcal{V}} = \mathcal{N} \setminus \mathcal{V}$, the cut-set $\partial \mathcal{V}$ is the set of all links that connect nodes in $\mathcal{V}$ to nodes in $\bar{\mathcal{V}}$. In other words:
\begin{align*}
\partial \mathcal{V} = \left\{ (i, j) \in \mathcal{L} \big| i \in \mathcal{V}, j \in \bar{\mathcal{V}} \right\}.
\end{align*}
\end{definition}
The cut size of a cut-set $\partial \mathcal{V}$ equals the number of links in the cut-set and is denoted as $|\partial \mathcal{V}|$. The maximum cut problem and the corresponding decision problem are as follows.
\begin{definition}[Maximum Cut Problem]\label{def:maxcut}
Given a graph $G$, find a cut $\partial \mathcal{V}$ of maximal cut size $|\partial \mathcal{V}|$.
\end{definition}
\begin{definition}[Maximum Cut Decision Problem]\label{def:maxcutdecision}
Given a natural number $\kappa$ and a graph $G$, is there a cut $\partial \mathcal{V}$ such that its cut size $|\partial \mathcal{V}|$ is at least $\kappa$? 
\end{definition}

 The maximum cut decision problem is NP-complete, as shown by Garey \textit{et al.} \cite{garey1976some}. Hence, the maximum cut problem is NP-hard \cite{cormen2009introduction}. The maximum cut problem can be equivalently stated as zero-one unconstrained quadratic programming (UQP) \cite{caprara2008constrained}
 \begin{align} \label{01uqp}
 \begin{aligned}
&\underset{y_1, ..., y_N}{\text{maximise }} & &  \sum^N_{i = 1}\sum^N_{j = i+1} a_{ij} (y_i(1-y_j)+y_j(1-y_i))& \\
 &\text{subject to } & &  y_i \in \{0, 1\}, \quad i= 1, ..., N. &
\end{aligned}
 \end{align}
The binary variable $y_i$ equals 1 if node $i$ is in the node set $\mathcal{V}$, and $y_i= 0$ if node $i$ is in the node set $\bar{\mathcal{V}}$. The optimisation problem (\ref{01uqp}) is equivalent to
 \begin{align}\label{01uqpsimpler}
 \begin{aligned}
&\underset{y_1, ..., y_N}{\text{maximise }} & &  \sum^N_{i = 1} \sum^N_{j = i+1} b_{ij} y_i y_j + \sum^N_{l = 1} b_{l} y_l&  \\
 &\text{subject to } & &  y_i \in \{0, 1\}, \quad i= 1, ..., N.&
\end{aligned}
 \end{align}
The coefficients of the objective function of (\ref{01uqpsimpler}) are given by 
\begin{align} \label{condcij}
b_{ij} = -2a_{ij} 
\end{align}
and the degree of node $l$
\begin{align}\label{condcl}
b_l = \sum^{N}_{j = 1} a_{lj}.
\end{align}
Since the elements $a_{ij}$ of the adjacency matrix $A$ are either zero or one, the coefficients are in the sets
\begin{align} \label{valuescij}
b_{ij} &\in \{ -2, 0\}
\end{align}
and
\begin{align} 
b_l &\in \{ 0, 1, ..., N-1\}.\label{valuescl}
\end{align}
The objective function $f$ of the optimisation problem (\ref{01uqpsimpler}) is a quadratic function which maps $N$ binary variables to a non-negative integer, i.e. $f: \{0, 1\}^N \mapsto \mathbb{N}_0$. Hence, the optimisation problem (\ref{01uqpsimpler}) is a special case of pseudo-Boolean optimisation \cite{boros2002pseudo}, in which the objective function $f$ maps $N$ binary variables to a real number, i.e. $f: \{0, 1\}^N \mapsto \mathbb{R}$. Rosenberg \cite{rosenberg1975reduction} showed that the optimisation of any pseudo-Boolean function can always be reduced in polynomial time to the optimisation of a quadratic pseudo-Boolean function. The general optimisation of a quadratic pseudo-Boolean function is of the form (\ref{01uqpsimpler}) with the difference that the coefficients $b_{ij}$ and $b_{l}$ may attain any value in $\mathbb{R}$ - not only the integer values in (\ref{valuescij}) and (\ref{valuescl}) - and is NP-hard \cite{garey1979guide}. If the coefficients $b_{ij}$ are non-negative real numbers, then the zero-one UQP (\ref{01uqpsimpler}) is polynomially solvable \cite{picard1975minimum}. There are other special cases for the range of values of the coefficients $b_{ij}$ and $b_{l}$ for which the zero-one UQP (\ref{01uqpsimpler}) is solvable in polynomial time \cite {pardalos1991graph, barahona1986solvable}. 

\section{Reduction of Maximum Cut to SIS Network Reconstruction}
\label{sec:reductionmaxcut}

We will show that any instance of the zero-one UQP (\ref{01uqpsimpler}) with coefficients $b_{ij}$ and $b_{l}$ in the sets (\ref{valuescij}) and (\ref{valuescl}), and thus any instance of the maximum cut problem, can be translated to an SIS network reconstruction problem (\ref{maxlikelihood}) in polynomial time. Hence, the SIS network reconstruction (\ref{maxlikelihood}) is NP-hard. Since the zero-one UQP (\ref{01uqpsimpler}) is not NP-hard for certain ranges \cite{picard1975minimum, pardalos1991graph, barahona1986solvable} of values of the coefficients $b_{ij}$ and $b_l$, we emphasise that the conditions (\ref{valuescij}) and (\ref{valuescl}) are crucial (at least sufficient) for the NP-hardness of the zero-one UQP (\ref{01uqpsimpler}). Thus, our aim is to show that the SIS network reconstruction problem (\ref{maxlikelihood}) can be translated to a zero-one UQP (\ref{01uqpsimpler}) with \textit{any}\footnote{More precisely, the coefficients $b_{ij}$ and $b_l$ do not attain any values in $\{ -2, 0\}$ and $\{ 0, 1, ..., N-1\}$ independently. Due to (\ref{condcij}) and (\ref{condcl}), it holds $b_l = -\frac{1}{2} \sum^{N}_{j = 1} b_{lj}$. We show the \textit{stronger} statement that, independently of the coefficients $b_{ij}$, the coefficients $b_{l}$ may attain any value in $\{ 0, 1, ..., N-1\}$.} coefficients $b_{ij}$ and $b_l$ in the sets given by (\ref{valuescij}) and (\ref{valuescl}). Since the SIS network reconstruction problem (\ref{maxlikelihood}) is fully specified by the viral state observations $x[1], ..., x[n]$, we aim to find viral state transitions $x[1], ..., x[n]$ such that solving the SIS network reconstruction problem (\ref{maxlikelihood}) is equivalent to solving the zero-one UQP (\ref{01uqpsimpler}). The proof of the NP-hardness of the SIS network reconstruction problem (\ref{maxlikelihood}) is based on four lemmas, which are stated below and whose proofs are given in the Appendix. 

 Since a graph $G$ given by an adjacency matrix $A$ in $\mathcal{A}$ is connected, there is a node $l$ such that the graph $G$ remains connected if node $l$ is removed: Indeed, in any connected graph, there exists a spanning tree that connects all the nodes. In any tree, there exists a node $l$ with degree one (a leaf node), whose removal does not disconnect the spanning tree and hence neither the graph. Without loss of generality, we label this node $l$ as node 1.
 
Our approach is based on stating a reduced-size version of the ML estimation (\ref{maxlikelihood}), namely only with respect to the links $a_{1i}$ which are incident to node 1. Since the graph given by an adjacency matrix $A$ in $\mathcal{A}$ is connected, node 1 has at least one neighbour. Without loss of generality, we label this neighbour as node 2. Furthermore, we consider that $a_{12}=1$ is known. In the following, we abbreviate 
\begin{align*}
\mathrm{Pr} \left[x[1], ..., x[n] \Big| \hat{a}_{13}, ..., \hat{a}_{1N}, \hat{a}_{12}= a_{12}, \hat{a}_{ij} = a_{ij} ~ \forall i,j \ge 2 \right],
\end{align*}
i.e. the likelihood when the elements $\hat{a}_{12}$ and $\hat{a}_{ij}$ for $i, j\ge 2$ are fixed to the true values, formally by 
\begin{align*}
\mathrm{Pr} \left[x[1], ..., x[n] \Big| \hat{a}_{13}, ..., \hat{a}_{1N} \right],
\end{align*}
and we introduce the following reduced-size SIS network estimation problem: 
\begin{definition}[Reduced-Size SIS Network Reconstruction]\label{def:redsizenetrecon}
Given the links $a_{12} = 1$ and $a_{ij}$, where $i \ge 2$ and $j \ge 2$, of the matrix $A \in \mathcal{A}$ and the viral state observations $x[k] \in \{0, 1\}^N$ from time $k=1$ to time $k=n$, which resulted from a sampled-time SIS process with the adjacency matrix $A$, find the links $(A_\textup{\textrm{ML}})_{13}, ..., (A_\textup{\textrm{ML}})_{1N}$ which maximise the log-likelihood:
\begin{align} \label{maxlikelireduced}
 \begin{aligned}
\left((A_\textup{\textrm{ML}})_{13}, ..., (A_\textup{\textrm{ML}})_{1N} \right) =&\underset{\hat{a}_{13}, ..., \hat{a}_{1N}}{\textup{ arg max }} & &  \log \left(\mathrm{Pr} \left[x[1], ..., x[n] \Big| \hat{a}_{13}, ..., \hat{a}_{1N} \right] \right) & \\
 &\textup{ s.t. } & &  \hat{a}_{1i} \in \{0, 1\}, \quad i= 3, ..., N.&
\end{aligned}
\end{align}
\end{definition}
Lemma \ref{lemma:sis_as_UQP} states that solving the reduced-size SIS network reconstruction (\ref{maxlikelireduced}) is equivalent to solving a zero-one UQP with particular coefficients:
\begin{lemma}[Reduced-Size SIS Network Reconstruction as Zero-One UQP]\label{lemma:sis_as_UQP}
For some natural numbers $m_0$, $m_{1l}$, $m_{2l} \in \mathbb{N}$, $l \in \{3, ..., N\}$, define the coefficients
 \begin{align}
 c_{ij} &\in \{-2, 0\}, \quad \quad \quad \quad \quad \quad \quad  i,j = 3, ..., N, \label{cijLEmm}\\
 c_l &=   \frac{m_{1l}}{m_0}  \lambda_{+}   + \frac{m_{2l}}{m_0}  \lambda_{-}  + \eta_l,\quad \quad l = 3, ..., N, \label{clLEmm}
 \end{align}  
where $\lambda_{+} > 0 $, $\lambda_{-}<0$ and $\eta_l \ge 0$ are constant and are given by the equations (\ref{lambdaPlus}), (\ref{lambdaMinus}) and (\ref{etaL}), respectively. For any coefficients $c_{ij}$ and $c_l$ given by (\ref{cijLEmm}) and (\ref{clLEmm}) and for any connected adjacency matrix $A\in \mathcal{A}$, there is a viral state sequence $x[k]$ from time $k=1$ to a finite time $k=n$ such that the reduced-size SIS network reconstruction problem (\ref{maxlikelireduced}) becomes:
\begin{align}\label{opt_prob_ooo}
 \begin{aligned}
&\underset{\hat{a}_{13}, ..., \hat{a}_{1N}}{\textup{max }} & & \sum^{N}_{i = 3} \sum^{N}_{j = i+1} c_{ij} \hat{a}_{1i} \hat{a}_{1j} + \sum^{N}_{l = 3} c_l \hat{a}_{1l} & \\
 &\textup{s.t. } & &  \hat{a}_{1i} \in \{0, 1\}, \quad i= 3, ..., N&
\end{aligned}
\end{align}
\end{lemma}
\begin{proof}
Appendix \ref{appendix:sis_as_UQP}.
\end{proof} 
Comparing the objective function of (\ref{opt_prob_ooo}) to the objective function in the zero-one UQP (\ref{01uqpsimpler}) shows that they are of the same form\footnote{The reduced-size SIS network reconstruction (\ref{maxlikelireduced}) for a graph with $N$ nodes results in a zero-one UQP (\ref{01uqpsimpler}) with $N-2$ optimisation variables $\hat{a}_{13}, ..., \hat{a}_{1N}$. Strictly speaking, to obtain the zero-one UQP (\ref{01uqpsimpler}) with $N$ optimisation variables, one has to consider the reduced-size SIS network reconstruction (\ref{maxlikelireduced}) for graphs with $N+2$ nodes. For ease of exposition, we omit the detail of the deviation of the number of optimisation variables of the two  optimisation problems (\ref{01uqpsimpler}) and (\ref{maxlikelireduced}).}: the binary variables $y_j$ in (\ref{01uqpsimpler}) correspond to $\hat{a}_{1j}$, and the coefficients $b_{ij}$ and $b_l$ in (\ref{01uqpsimpler}) are replaced by $c_{ij}$ and $c_l$ in (\ref{opt_prob_ooo}), respectively.

As stated in the beginning of Section \ref{sec:reductionmaxcut}, a crucial condition for the NP-hardness of the zero-one UQP (\ref{01uqpsimpler}) is that its coefficients are in the sets $b_{ij} \in \{-2, 0\}$ and $b_l \in \{0, 1, ..., N-1\}$. To show the NP-hardness of the zero-one UQP (\ref{opt_prob_ooo}), we have to show that also the coefficients $c_{ij}$ and $c_l$ attain \textit{any} value in $\{-2, 0\}$ and $\{0, 1, ..., N-1\}$, respectively. As stated by (\ref{cijLEmm}), the coefficients $c_{ij}$ may attain either value in $\{-2, 0\}$. The remaining condition that the coefficients $c_l$, given by (\ref{clLEmm}), may attain any value in $\{0, 1, ..., N-1\}$ \textit{exactly} does generally not hold. Nevertheless, the coefficients $c_l$ may approach any $b_l \in \{0, 1, ..., N-1\}$ \textit{arbitrarily close}, as stated by Lemma \ref{lemma:approaching}.
\begin{lemma}[Coefficients Approach Any Number]\label{lemma:approaching}
The coefficients $c_l$ of the optimisation problem (\ref{opt_prob_ooo}), given by (\ref{clLEmm}), may approach any numbers $b_l \in \mathbb{R}$, $l = 3, ..., N$, arbitrarily close for suitably chosen natural numbers $m_0, m_{1l}, m_{2l} \in \mathbb{N}$:
\begin{align}\label{epsiloncloseness}
\forall\varepsilon \in \mathbb{R}^+, l\in \{3, ..., N\}, z_l \in \mathbb{R}: ~~ \exists m_0, m_{1l}, m_{2l} \in \mathbb{N} \quad \text{such that} \quad |c_l -b_l | \le \varepsilon 
\end{align} 
\end{lemma} 
\begin{proof}
Appendix \ref{appendix:approaching}.
\end{proof} 

 If the deviation $(c_l - b_l)$ is positive and not greater than a threshold $\varepsilon = \frac{1}{N}$, then we can solve any instance of the maximum-cut problem by solving an instance of the reduced-size SIS network reconstruction (\ref{maxlikelireduced}):
\begin{lemma}[Sufficiently Small Error on the UQP Coefficients]\label{lemma:sufficientlyClose}
If $c_l \ge b_l$ and $c_l - b_l < \frac{1}{N}$ for all $l\in \{3, ..., N\}$, then the solution to the reduced-size SIS network reconstruction problem (\ref{opt_prob_ooo}) is also a solution to the zero-one UQP (\ref{01uqpsimpler}). 
\end{lemma} 
\begin{proof}
Appendix \ref{appendix:sufficientlyClose}.
\end{proof} 

Lemma \ref{lemma:sis_as_UQP}, Lemma \ref{lemma:approaching} and Lemma \ref{lemma:sufficientlyClose} prove the NP-hardness of the reduced-size SIS network reconstruction (\ref{maxlikelireduced}). Lemma \ref{lemma:redsizenetrecon} states how to obtain the reduced-size SIS network reconstruction (\ref{01uqpsimpler}) from the original, full-size SIS network reconstruction problem (\ref{maxlikelihood}).

\begin{lemma}[From Full-Size to Reduced-Size SIS Network Reconstruction]\label{lemma:redsizenetrecon}
For all connected adjacency matrices $A \in \mathcal{A}$ and all viral state sequence $x[1], ..., x[n_1]$, there is a viral state sequence $x[1], ..., x[n_2]$ with $n_2 > n_1$, such that the solution $A_\textup{\textrm{ML}}$ to the full-size SIS network reconstruction (\ref{maxlikelihood}) satisfies:
\begin{enumerate}
\item The following elements of $A_\textup{\textrm{ML}}$ equal the elements of the true adjacency matrix $A$:
\begin{align*}
 (A_\textup{\textrm{ML}})_{12} = a_{12}= 1 \quad  \text{and} \quad (A_\textup{\textrm{ML}})_{ij} = a_{ij} \quad \text{for all } i, j \ge 2
\end{align*}
\item The other elements of $A_\textup{\textrm{ML}}$ are the solution to the reduced-size SIS network reconstruction problem (\ref{maxlikelireduced}) whose objective function is changed by an additive term:
\begin{align} \label{redSizeSISAmend}
 \begin{aligned}
\left((A_\textup{\textrm{ML}})_{13}, ..., (A_\textup{\textrm{ML}})_{1N} \right)= &\underset{\hat{a}_{13}, ..., \hat{a}_{1N}}{\textup{ arg max }} & &  \log\left(\mathrm{Pr} \left[x[1], ..., x[n_1] \Big| \hat{a}_{13}, ..., \hat{a}_{1N} \right]\right)\\
&&& +\sum^N_{l=2} \kappa_l \log\left( 1 - \delta_T - \beta_T d_l(A) - \beta_T \hat{a}_{1l}\right)& \\
 &\textup{s.t. } & &  \hat{a}_{1l} \in \{0, 1\}, \quad l= 3, ..., N.&
\end{aligned}
\end{align}
Here, $d_l(A)= \sum^N_{m=2} a_{ml}$ denotes the degree of node $l$ when node 1 is removed from the graph given by the adjacency matrix $A$, and $\kappa_l$ is a natural number which is independent of the optimisation variables $\hat{a}_{13}, ..., \hat{a}_{1N}$.
\end{enumerate}
\end{lemma} 
\begin{proof}
Appendix \ref{appendix:redsizenetrecon}.
\end{proof} 
The optimisation problem (\ref{redSizeSISAmend}) resembles the reduced-size SIS network reconstruction (\ref{maxlikelireduced}), but the objective functions differ by the additive term $\sum^N_{l=2} \kappa_l \log\left( 1 - \delta_T - \beta_T d_l(A) - \beta_T \hat{a}_{1l}\right)$. We show in Appendix \ref{appendix:NPhard} that the additive term does not have an impact on the difficulty: The NP-hardness of the reduced-size SIS network reconstruction (\ref{maxlikelireduced}) implies the NP-hardness of the optimisation problem (\ref{redSizeSISAmend}). Since solving the full-size SIS network reconstruction problem (\ref{maxlikelihood}) with the viral state sequence $x[1], ..., x[n_2]$ as input implies solving the NP-hard optimisation problem (\ref{redSizeSISAmend}), we obtain the main theorem of this work: 
\begin{theorem}[SIS Network Reconstruction is NP-Hard] \label{theorem:NPhard}
For all connected adjacency matrices $A\in \mathcal{A}$, the SIS network reconstruction problem (\ref{maxlikelihood}) is NP-hard.
\end{theorem}
\begin{proof}
Appendix \ref{appendix:NPhard}.
\end{proof} 
We emphasise that the NP-hardness holds for \textit{any} class of connected adjacency matrices $A\in \mathcal{A}$, also for simple topologies such as paths or star graphs.
 
 \section{Conclusions} 

This work considers the computational complexity of finding the ML estimate of the network topology from observing a sampled-time SIS viral state trace. Instead of reconstructing a network for a \textit{given} viral state sequence, we considered the reverse problem of \textit{designing} a viral state sequence such that estimating the presence or absence of links either becomes computationally difficult (Lemma \ref{lemma:sis_as_UQP}) or easy (first statement of Lemma \ref{lemma:redsizenetrecon}). 

Specifically, we have shown that any instance of the NP-hard maximum cut problem can be reduced to an instance of the SIS network reconstruction problem, whereby an instance of the latter problem is given by a viral state sequence. Thus, we have proved that the ML network reconstruction for SIS processes is NP-hard. In general, the exact ML estimate of the network topology can hence not be computed in polynomial time. The NP-hardness is a worst case result, and we emphasise two points. Firstly, it may be possible that the ML network reconstruction can be solved for some classes of practical problems within a reasonable computation time. Nevertheless, it remains to study which viral state sequences could result (possibly on average) in a low computational complexity. Secondly, considering the inapproximability results for the maximum cut problem \cite{gartner2012approximation}, one might be tempted to conclude that an accurate reconstruction of the network for SIS processes is not possible in polynomial time. However, a thorough analysis of the accuracy of the exact ML estimator of an unweighted (and hence discrete valued) adjacency matrix $A$ is an open question.

\section*{Acknowledgements}

We are grateful to Jaron Sanders for helpful discussions on this material.

\bibliographystyle{ieeetr}                                                                                             

\appendix

\section{Proof of Lemma \ref{lemma:sis_as_UQP}}
\label{appendix:sis_as_UQP}
The objective function of (\ref{maxlikelireduced}) equals
\begin{align}
f_n(\hat{a}_{13}, ..., \hat{a}_{1N}) &= \log\left(\mathrm{Pr} \left[x[1], ..., x[n] \Big| \hat{a}_{13}, ..., \hat{a}_{1N} \right] \right) \nonumber\\
&= \sum^{n-1}_{k=1} \log\left(\mathrm{Pr} \left[x[k+1] \Big| x[k],\hat{a}_{13}, ..., \hat{a}_{1N} \right] \right),\label{loglikelihooood}
\end{align}
where the last equality follows from the Markov property of the sampled-time SIS process. To reduce the zero-one UQP (\ref{01uqpsimpler}) to the reduced-size SIS network reconstruction problem (\ref{maxlikelireduced}), we show below that it is possible to construct a series of viral state transitions $x[k] \rightarrow x[k+1]$ for the time points $k = 1, ..., n-1$ for all adjacency matrices $A \in \mathcal{A}$, such that the objective function $f_n$ of the latter problem is of the form
\begin{align}\label{fobj_UQP}
f_n(\hat{a}_{13}, ..., \hat{a}_{1N}) = \sum^{N}_{i = 3} \sum^{N}_{j = i+1} g_{ij} \hat{a}_{1i} \hat{a}_{1j} + \sum^{N}_{l = 3} g_{l} \hat{a}_{1l} + g_\textrm{const},
\end{align}
with the coefficients $g_{ij}$ and $g_l$ and an additive term $g_\textrm{const}$ which is constant with respect to the links $\hat{a}_{13}, ..., \hat{a}_{1N}$ and, hence, can be omitted in the optimisation problem (\ref{maxlikelireduced}). We prove Lemma \ref{lemma:sis_as_UQP} in five steps, on which we elaborate in detail in the respective Subsections \ref{subsec:trans1} to \ref{subsec:multiplicity}.
\begin{enumerate}
\item We design a viral state transition $\mathfrak{I}_{ij}: x[k] \rightarrow x[k+1]$ which results in setting the \textit{quadratic costs} $g_{ij}$ of (\ref{fobj_UQP}) to a value. In Subsection \ref{subsec:multiplicity}, we show that if the viral state transition $\mathfrak{I}_{ij}$ occurs, then we obtain $g_{ij} = -2$, and if it does not occur, then we obtain $g_{ij} = 0$.

\item We design a viral state transition $\mathfrak{I}_{l}: x[k] \rightarrow x[k+1]$ which results in setting the \textit{linear costs} $g_l$ of (\ref{fobj_UQP}) to a \textit{positive} value $g_l>0$. 

\item We design a viral state transition $\mathfrak{C}_{l}: x[k] \rightarrow x[k+1]$ which results in setting the linear cost $g_l$ of (\ref{fobj_UQP}) to a \textit{negative} value $g_l<0$. 

\item We show how two transitions of the kind $\mathfrak{I}_{ij}, \mathfrak{I}_{l}$ and $\mathfrak{C}_{l}$ can be connected by constructing a suitable transition sequence.

\item  We show that it is possible to construct a viral state sequence $x[1], ..., x[n]$ which is composed of several of the three kinds of viral state transitions $\mathfrak{I}_{ij}, \mathfrak{I}_{l}$ and $\mathfrak{C}_{l}$. If the viral state transition $\mathfrak{I}_{l}$ occurs multiple times, then the value of the coefficient $g_l$ increases. On the other hand, if the viral state transition $\mathfrak{C}_{l}$ occurs multiple times, then the value of the coefficient $g_l$ decreases\footnote{In the following Lemma \ref{lemma:approaching}, we show that the coefficient $g_l$ can be set (arbitrarily close) to any value in $\mathbb{R}$ by adjusting the number of occurrences of the transitions $\mathfrak{I}_{l}$ and $\mathfrak{C}_{l}$.}. By choosing the multiplicity of the occurrence of viral state transitions $\mathfrak{I}_{ij}$, $\mathfrak{I}_{l}$ and $\mathfrak{C}_{l}$, we show that the reduced-size SIS network reconstruction (\ref{maxlikelireduced}) becomes a zero-one UQP of the form (\ref{opt_prob_ooo}).
\end{enumerate}

\subsection{Setting the Quadratic Costs}
\label{subsec:trans1}
In order to set the coefficients $g_{ij}$ for $i\ge 3$ and $j \ge i+1$, corresponding to the terms $g_{ij} \hat{a}_{1i} \hat{a}_{1j}$ in the objective function (\ref{fobj_UQP}), we construct the following special case of an infectious transition (\ref{infectionTrans}). The links $\hat{a}_{1i}$ and $\hat{a}_{1j}$ appear simultaneously in the probability for the infectious transition (\ref{infectionTrans}) if both node $i$ and node $j$ are infected at time $k$, i.e. $x_i[k] = x_j[k] = 1$, and node 1 becomes infected at time $k+1$, i.e. $x_1[k]= 0 \rightarrow x_1[k+1]= 1$. We choose the viral state of node 2 as\footnote{If node $i$ and $j$ were the only infected nodes at time $k$, then the transition probability (\ref{transitionProblCij}) would equal zero if both elements $\hat{a}_{1i}= 0$ and $\hat{a}_{1j}= 0$. In that case, we would not be able to express the logarithm of the transition probability in the form (\ref{logTransCCCC}) for all values of the elements $\hat{a}_{1i}, \hat{a}_{1j} \in \{0, 1\}$.} $x_2[k] = 1$ and define the transition
\begin{align*}
\mathfrak{I}_{ij} = \left\{x[k+1] = e_1 +e_2 + e_i + e_j \big| x[k] = e_2 + e_i + e_j\right\}.
\end{align*}
The elements of the vector $e_i \in \mathbb{R}^N$ are given by $(e_i)_m = \delta_{mi}$, where $\delta_{mi}$ is the Kronecker delta. The transition $\mathfrak{I}_{ij}$ is a special case of an infectious transition (\ref{infectionTrans}) and, since $\hat{a}_{12} = a_{12} = 1$ in the reduced-size SIS network reconstruction (\ref{maxlikelireduced}), its transition probability is given by
 \begin{align}\label{transitionProblCij}
 \mathrm{Pr} \left[ \mathfrak{I}_{ij} \Big|  \hat{a}_{13}, ..., \hat{a}_{1N} \right] = \begin{cases} \beta_T \quad ~ &\text{if} \quad \hat{a}_{1i} = 0 \land \hat{a}_{1j} = 0,\\
 2 \beta_T  &\text{if} \quad (\hat{a}_{1i} = 0 \land \hat{a}_{1j} = 1) \lor (\hat{a}_{1i} = 1 \land \hat{a}_{1j} = 0), \\
 3 \beta_T  &\text{if} \quad \hat{a}_{1i} = 1 \land \hat{a}_{1j} = 1.
 \end{cases}
 \end{align}
To compute the objective function $f_n$ according to (\ref{loglikelihooood}), we express the logarithm of the above transition probability (\ref{transitionProblCij}) more compactly as
 \begin{align}
\log\left( \mathrm{Pr} \left[\mathfrak{I}_{ij} \Big|    \hat{a}_{13}, ..., \hat{a}_{1N} \right]\right) &= (1- \hat{a}_{1i})(1 - \hat{a}_{1j}) \log(\beta_T) + \hat{a}_{1i}(1 - \hat{a}_{1j}) \log(2 \beta_T)   \nonumber\\
 & \quad+ (1- \hat{a}_{1i})\hat{a}_{1j} \log(2 \beta_T) + \hat{a}_{1i}\hat{a}_{1j} \log(3\beta_T)\nonumber\\
 &= \log(\beta_T) + \hat{a}_{1i} \log(2)+ \hat{a}_{1j} \log(2) +\hat{a}_{1i}\hat{a}_{1j}\log\left(\frac{3}{4}\right) \label{logTransCCCC}
 \end{align} 
 If solely the transition $\mathfrak{I}_{ij}$ occurred once, then it follows from (\ref{logTransCCCC}) that the quadratic cost of (\ref{fobj_UQP}) would equal $g_{ij} = \log\left(\frac{3}{4}\right)<0$. We emphasise that the transitions $\mathfrak{I}_{ij}$ only need to occur for $i\ge 3$ and $j \ge i+1$ since the quadratic coefficients $g_{ij}$ in the objective function (\ref{fobj_UQP}) only occur for those values of $i$ and $j$.

\subsection{Setting the Linear Costs to a Positive Value}
\label{subsec:trans2}
In order to set the coefficients $g_{l}$, corresponding to the terms $g_{l} \hat{a}_{1l}$ in the objective function of (\ref{fobj_UQP}), to a positive value $g_{l}>0$, we construct the following special case of an infectious transition (\ref{infectionTrans}). The link $\hat{a}_{1l}$ appears in the probability for the infectious transition (\ref{infectionTrans}) if node $l$ is infected at time $k$, i.e. $x_l[k] = 1$, and node 1 becomes infected at time $k+1$, i.e. $x_1[k]= 0 \rightarrow x_1[k+1]= 1$. Analogously to Subsection \ref{subsec:trans1}, we choose the viral state of node 2 as $x_2[k] = 1$ and define the transition
\begin{align*}
\mathfrak{I}_{l}= \left\{x[k+1] = e_1 +e_2 + e_l \big| x[k] = e_2 + e_l\right\}.
\end{align*}
The transition $\mathfrak{I}_{l}$ is a special case of an infectious transition (\ref{infectionTrans}). Since $\hat{a}_{12} = a_{12} = 1$ in the reduced-size SIS network reconstruction (\ref{maxlikelireduced}), the transition probability of $\mathfrak{I}_{l}$ is given by
 \begin{align}\label{kjbkjbkjbkjbk}
 \mathrm{Pr} \left[ \mathfrak{I}_{l} \Big| \hat{a}_{13}, ..., \hat{a}_{1N} \right] = \begin{cases} \beta_T \quad ~ &\text{if} \quad \hat{a}_{1l} = 0,\\
 2 \beta_T  &\text{if} \quad \hat{a}_{1l} = 1.
 \end{cases}
 \end{align}
To compute the objective function $f_n$ according to (\ref{loglikelihooood}), we obtain the logarithm of the above transition probability (\ref{kjbkjbkjbkjbk}) as
 \begin{align}
 \log\left(\mathrm{Pr} \left[\mathfrak{I}_{l} \Big|   \hat{a}_{13}, ..., \hat{a}_{1N} \right]\right) &= (1- \hat{a}_{1l}) \log(\beta_T) + \hat{a}_{1l} \log(2 \beta_T) \nonumber\\
 &=\log(\beta_T)  +  \hat{a}_{1l} \log(2).\label{transitionTl}
 \end{align}
 If solely the transition $\mathfrak{I}_{l}$ occurred once, then it follows from (\ref{transitionTl}) that the linear cost of (\ref{fobj_UQP}) would equal $g_l =  \log(2)>0$.
 
\subsection{Setting the Linear Costs to a Negative Value}
\label{subsec:trans3}
In order to set the coefficients $g_{l}$, corresponding to the terms $g_{l} \hat{a}_{1l}$ in the objective function of (\ref{fobj_UQP}), to a negative value $g_{l}<0$, we construct the following special case of a constant transition (\ref{constTrans}). The link $\hat{a}_{1l}$ appears in the probability for the constant transition (\ref{constTrans}) if node 1 is susceptible and node $l$ is infected ($x_1[k]=0$ and $x_l[k]= 1$). Hence, we define the transition
\begin{align}\label{clTransition}
\mathfrak{C}_{l}= \left\{x[k+1] = e_l \big| x[k] = e_l\right\}.
\end{align}
The transition $\mathfrak{C}_{l}$ is a special case of a constant transition (\ref{constTrans}) and its transition probability can be calculated as follows. From time $k$ to time $k+1$, the probability of the infection of a node $m \neq l$ is
\begin{align*}
\mathrm{Pr}\left[ \text{Node $m$ gets infected at } k+1\Big| x[k]=e_l,  \hat{a}_{13}, ..., \hat{a}_{1N} \right] &= \beta_T \hat{a}_{ml}
\end{align*}
 The probability of an infection of a node at the time $k+1$ is hence
 \begin{align*}
\mathrm{Pr}\left[ \text{A node gets infected at } k+1\Big| x[k]=e_l,  \hat{a}_{13}, ..., \hat{a}_{1N} \right] &= \sum^{N}_{m = 1, m \neq l}\beta_T \hat{a}_{ml}  \\
& = \beta_T \hat{a}_{1l} + \beta_T+\beta_T  \sum^{N}_{m = 3, m \neq l} \hat{a}_{ml}, 
\end{align*}
since $\hat{a}_{12} = a_{12}= 1$ in the reduced-size SIS network reconstruction (\ref{maxlikelireduced}). The probability of the curing (\ref{healTransition}) of node $l$ equals $\delta_T$. Thus, the probability for the constant transition (\ref{clTransition}) becomes
\begin{align} 
\mathrm{Pr} \left[ \mathfrak{C}_{l} \Big|  \hat{a}_{13}, ..., \hat{a}_{1N} \right]  &= 1 - \delta_T -\beta_T \hat{a}_{1l} -\beta_T - \beta_T \sum^{N}_{m = 3, m \neq l} \hat{a}_{ml} \nonumber\\
&= \xi-\beta_T \hat{a}_{1l},\label{lllllqqwww}
\end{align}
where 
\begin{align*} 
\xi=   1 - \delta_T -\beta_T - \beta_T \sum^{N}_{m = 3, m \neq l} \hat{a}_{ml} 
\end{align*}
is constant with respect to the links $\hat{a}_{13}, ..., \hat{a}_{1N}$ and does not have to be considered in the optimisation problem (\ref{maxlikelireduced}). It holds that $\mathrm{Pr} \left[ \mathfrak{C}_{l} \Big|  \hat{a}_{13}, ..., \hat{a}_{1N} \right]$ is in $[0, 1]$ for all link estimates $\hat{a}_{13}, ..., \hat{a}_{1N}$, which implies that $\xi > 0$. To compute the objective function $f_n$ according to (\ref{loglikelihooood}), we obtain the logarithm of the transition probability (\ref{lllllqqwww}) as
 \begin{align}
 \log\left(\mathrm{Pr} \left[ \mathfrak{C}_{l} \Big| \hat{a}_{13}, ..., \hat{a}_{1N} \right]\right) &= (1 - \hat{a}_{1l} )\log\left(\xi\right) + \hat{a}_{1l}  \log\left(\xi - \beta_T\right)  \nonumber\\
 &= \log\left(\xi\right) + \hat{a}_{1l}  \log\left(1 - \frac{\beta_T}{\xi}\right). \label{logTransPosValue}
 \end{align}  
 If solely the transition $\mathfrak{C}_{l}$ occurred once, then it follows from (\ref{logTransPosValue}) that the linear cost of (\ref{fobj_UQP}) would equal $g_l = \log\left(1 - \frac{\beta_T}{\xi}\right) <0$.
 
 \subsection{Connecting Viral State Transitions}
\label{subsec:connecting}
In order to set the coefficients $g_l$ and $g_{ij}$ for more than one node $l$ (or for more than one pair of nodes $i$ and $j$), the transitions $\mathfrak{I}_{ij}$, $\mathfrak{I}_l$ and $\mathfrak{C}_l$ must occur multiple times in the viral state sequence $x[1], ..., x[n]$ for different values of $l$, $i$ and $j$. Consider that one of the transitions $\mathfrak{I}_{ij}$, $\mathfrak{I}_l$ or $\mathfrak{C}_l$ occurs from time $k_0$ to $k_0+1$ and that another (not necessarily different) of the transitions $\mathfrak{I}_{ij}$, $\mathfrak{I}_l$ or $\mathfrak{C}_l$ shall occur from time $k_0 +\Delta k$ to $k_0 + \Delta k +1$ for some $\Delta k\ge 1$. For \textit{any} connected adjacency matrix $A \in \mathcal{A}$, there is a viral state sequence which transform the viral state $x[k_0+1]$ at the end of one transition to the viral state $x[k_0+ \Delta k]$ at the beginning of another transition, as we show in the three steps below. 

\begin{enumerate}
\item If the transition $x[k_0] \rightarrow x[k_0+1]$ is one of the infectious transition $\mathfrak{I}_{ij}$ or $\mathfrak{I}_l$, then node 1 is infected at time $k_0+1$. In that case, we consider that node 1 cures from time $k_0+1$ to $k_0+2$. In the two steps below, replace formally time $k_0+1$ by $k_0+2$.

\item The expressions (\ref{logTransCCCC}), (\ref{transitionTl}) and (\ref{logTransPosValue}) influence the values of the coefficients $g_l$ and $g_{ij}$ in the objective function (\ref{fobj_UQP}). In order to give explicit expressions for coefficients $g_l$ and $g_{ij}$, we would like to achieve that the viral state transitions from time $k_0+1$ to $k_0+ \Delta k$ do not have an influence on the values of any of the coefficients $g_l$ and $g_{ij}$, such that their value is solely determined by the expressions (\ref{logTransCCCC}), (\ref{transitionTl}) and (\ref{logTransPosValue}). 

The coefficients $g_l$ and $g_{ij}$ correspond to addends in the objective function (\ref{fobj_UQP}), which include the links $\hat{a}_{1l}$, $\hat{a}_{1i}$ and $\hat{a}_{1j}$, which are incident to node 1. A link $\hat{a}_{1l}$, which is incident to node 1, appears in the expressions for the probability of a viral state transition $x[k] \rightarrow x[k+1]$ of the sampled-time SIS process for exactly two cases. Firstly, in the probability of an infectious transition (\ref{infectionTrans}) from time $k$ to $k+1$ only if node 1 is infected before or afterwards ($x_1[k]= 1$ or $x_1[k+1]= 1$). Secondly, the link $\hat{a}_{1l}$ may appear in the probability of a constant transition (\ref{constTrans}) from time $k$ to $k+1$. We thus would like to exclude these two kinds of transitions from time $k_0+1$ to $k_0+\Delta k$.

Hence, we want to construct the viral state transitions from time $k_0+1$ to $k_0 + \Delta k$ such that the first node is constantly susceptible ($x_1[k]= 0$ for $k= k_0+1, ..., k_0 + \Delta k$) and additionally, such that there is no constant transition (\ref{constTrans}) from time $k_0+1$ to $k_0 + \Delta k$. Then, the coefficients $g_l$ and $g_{ij}$ in the objective function (\ref{fobj_UQP}) are not affected by any of the viral state transitions from time $k_0+1$ to $k_0 + \Delta k$ and are solely determined by the expressions (\ref{logTransCCCC}), (\ref{transitionTl}) and (\ref{logTransPosValue}).

\item The graph given by an adjacency matrix $A \in \mathcal{A}$ remains connected if node 1 is removed as stated above Definition \ref{def:redsizenetrecon}. Thus, there exists a time $k_0 + \Delta k \ge k_0+1$ and a finite sequence of non-constant transitions of the SIS process which transforms the viral state $x[k_0+1]$ to any other viral state $x[k_1] \in \{0,1\}^{N-1}$ under the constraint that node 1 is susceptible $x_1[k]=0$ for time $k = k_0+1$ to $k_1$: The simplest of such transition sequences would be successive infections (\ref{infectionTrans}), resulting in all nodes $2, ..., N$ being infected, with a subsequent curing (\ref{healTransition}) of those nodes $i$ for which $x_i[k_0 + \Delta k] = 0$ shall hold.
\end{enumerate}

For a network of six nodes, Figure \ref{fig_transition} gives an illustration on how two infectious transitions, namely $\mathfrak{I}_{34}$ and $\mathfrak{I}_6$, can be connected by the viral state sequence described in the three steps above.

\begin{figure}[h!]
	 \includegraphics[width=\textwidth]{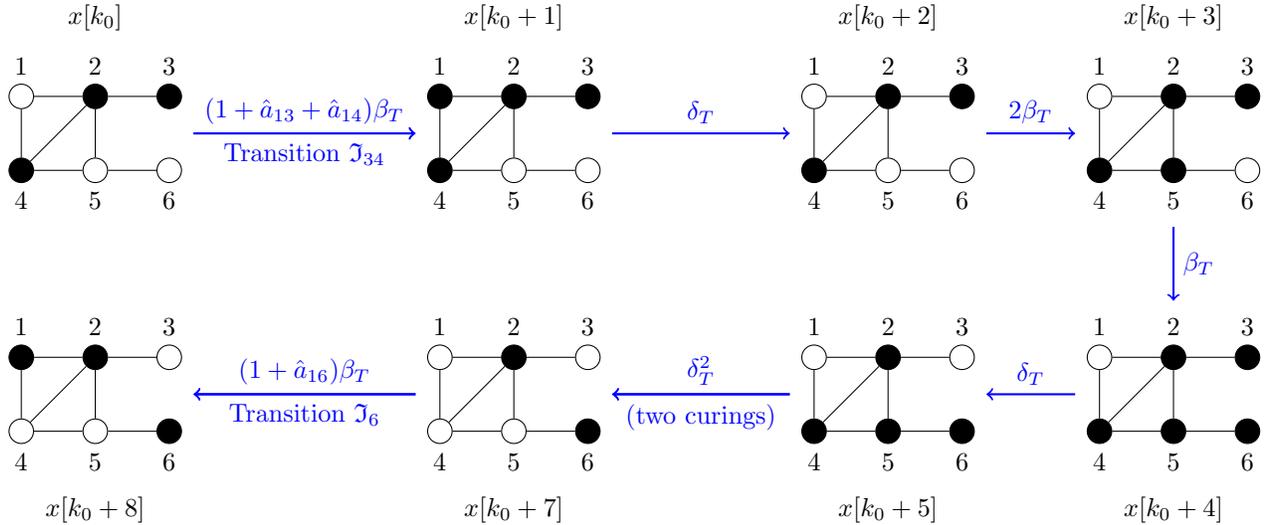}
	{\footnotesize\caption{	An illustration of connecting two viral transitions, namely $\mathfrak{I}_{34}$ from time $k_0$ to $k_0+1$ and $\mathfrak{I}_6$ from time $k_0+7$ to $k_0+8$, for a connected network of six nodes by the procedure described in Subsection \ref{subsec:connecting}. Above the blue arrows, the respective transition probabilities are stated. It holds $\hat{a}_{12}= a_{12}=1$ and $\hat{a}_{ij}= a_{ij}$ for $i, j \ge 2$ in the optimisation problem (\ref{maxlikelireduced}), and thus the transition probabilities from time $k_0+1$ to $k_0+7$ can be stated without the dependency on $\hat{a}_{ij}$. On the other hand, both transitions $\mathfrak{I}_{34}$ and $\mathfrak{I}_6$ do depend on the elements $\hat{a}_{ij}$. Since the transition $\mathfrak{I}_{34}$ from time $k_0$ to $k_0+1$ is an infectious transition, we consider that node 1 cures from time $k_0+1$ to $k_0+2$ according to step one in Subsection \ref{subsec:connecting}. Then, following the description in step two and three of Subsection \ref{subsec:connecting}, every node except node 1 becomes infected from time $k_0+1$ to $k_0+4$. Subsequently, the nodes 3, 4 and 5 cure from time $k_0+4$ to $k_0+7$ as required for the first state of the transition $\mathfrak{I}_6$. In Subsection \ref{subsec:multiplicity}, the viral state sequence from time $k_0+1$ to $k_0+7$ is also denoted by $\mathfrak{F}(x[k_0+1], \mathfrak{I}_6)$, and its length is given by $\tau\left(\mathfrak{F}(x[k_0+1], \mathfrak{I}_6)\right)= 6$.\label{fig_transition}}}
\end{figure}

\subsection{Constructing the Complete Viral State Sequence}
\label{subsec:multiplicity}
We consider that each of the viral state transitions $\mathfrak{I}_{ij}$, $\mathfrak{I}_l$ and $\mathfrak{C}_l$ may occur multiple times, and denote the multiplicities by $m_{ij}$, $m_{1l}$ and $m_{2l}$, respectively. By $\mathfrak{T}$ we denote a viral state transition that is of the kind $\mathfrak{I}_{ij}$, $\mathfrak{I}_l$ or $\mathfrak{C}_l$. Furthermore, we denote by $\mathfrak{F}\left( x[k], \mathfrak{T}\right)$ the viral state sequence which transforms the viral state $x[k]$ at time $k$ to the first state of the transition $\mathfrak{T}$ (see also Figure \ref{fig_transition} for an example), following the description in Subsection \ref{subsec:connecting}. The length (number of discrete time steps) of the viral state sequence $\mathfrak{F}\left( x[k], \mathfrak{T}\right)$ is denoted by $\tau\left(\mathfrak{F}\left( x[k], \mathfrak{T}\right)\right)$. The construction of the whole viral state sequence, which includes the viral state transitions $\mathfrak{I}_{ij}$, $\mathfrak{I}_l$ and $\mathfrak{C}_l$ with the multiplicities $m_{ij}$, $m_{1l}$ and $m_{2l}$, is given in pseudo-code by Algorithm \ref{alg:exploration_walk_one}. We emphasise that if a non-zero multiplicity of a viral state transition is increased, then only the respective for-loop (e.g. line 18 to line 21 for the transition $\mathfrak{I}_l$ if its multiplicity $m_{1l} =c \neq 0$ is increased to $m_{1l}=2c$ for some $c \in \mathbb{N}$) in Algorithm \ref{alg:exploration_walk_one} is run more often. In particular, line 9 is not executed more often when the multiplicities of the viral state transitions are increased.

\begin{algorithm}
\caption{Construction of Viral State Sequence for Reduced-Size SIS Network Reconstruction}
\begin{algorithmic}[1]
\State \textbf{Input: } graph $G = (\mathcal{N}, \mathcal{L})$, initial state $x[0]$, multiplicities $m_{1l}, m_{2l}, m_{ij}$ 
\State \textbf{Output: } viral state sequence $x[1], ..., x[n]$
\State $\mathcal{Q} \gets \left\{ \mathfrak{I}_{ij} \big| r_{ij} =1 \right\} \cup \left\{ \mathfrak{I}_{l} \big| m_{1l} \ge 1 \right\} \cup \left\{ \mathfrak{C}_{l} \big| m_{2l} \ge 1 \right\}$ \Comment{initialise the queue $\mathcal{Q}$}
\State $k \gets 1$
\While{$\mathcal{Q} \neq \emptyset$} 
\State $\mathfrak{T} \gets$ some element of $\mathcal{Q}$ \Comment{dequeue a transition $\mathfrak{T}$ from $\mathcal{Q}$}
\State $\mathcal{Q} \gets \mathcal{Q} \setminus \{\mathfrak{T}\}$ 
\State $\Delta k \gets \tau\left(\mathfrak{F}\left( x[k-1], \mathfrak{T}\right)\right)$  \Comment{length of transition from $x[k-1]$ to first state of $\mathfrak{T}$}
\State $(x[k], ..., x\left[k+\Delta k-1\right]) \gets \mathfrak{F}\left( x[k-1], \mathfrak{T}\right)$
\State $k \gets k+\Delta k$
\If{$\mathfrak{T}= \mathfrak{I}_{ij}$ for some $(i,j)$}
\For{$c=1, ..., ,m_{ij}$}
\State $(x[k], x[k+1])\gets (e_2 + e_i + e_j, e_1 + e_2 + e_i + e_j)$  \Comment{transition $\mathfrak{I}_{ij}$}
\State $k \gets k+2$
\EndFor
\EndIf
\If{$\mathfrak{T}= \mathfrak{I}_l$ for some $l$}
\For{$c=1, ..., ,m_{1l}$}
\State $(x[k], x[k+1])\gets (e_2+e_l, e_1+e_2+e_l)$ \Comment{transition $\mathfrak{I}_l$}
\State $k \gets k+2$
\EndFor
\EndIf
\If{$\mathfrak{T}= \mathfrak{C}_l$ for some $l$}
\For{$c=1, ..., ,m_{2l}$}
\State $x[k]\gets e_l$ \Comment{transition $\mathfrak{C}_l$}
\State $k \gets k+1$
\EndFor
\EndIf
\EndWhile
\State $n \gets k -1$
\end{algorithmic}
\label{alg:exploration_walk_one}
\end{algorithm}

In the following, we show how the multiplicities $m_{ij}$, $m_{1l}$ and $m_{2l}$ of the viral state transitions can be adjusted such that the reduced-size SIS network reconstruction (\ref{opt_prob_ooo}) attains the form (\ref{maxlikelireduced}). For the viral state sequence $x[1], ..., x[n]$ given by the output of Algorithm \ref{alg:exploration_walk_one}, the coefficients $g_{ij}$ of the objective function (\ref{fobj_UQP}) follow from the expression (\ref{logTransCCCC}) for the probability of the viral state transition $\mathfrak{I}_{ij}$ as
 \begin{align}
 g_{ij} = \log\left(\frac{3}{4}\right)  m_{ij}. \label{dij_first}
 \end{align}
Furthermore, the expressions (\ref{logTransCCCC}), (\ref{transitionTl}) and (\ref{logTransPosValue}) for the viral state transitions $\mathfrak{I}_{ij}$ (for $i\ge 3$ and $j \ge i+1$), $\mathfrak{I}_l$ and $\mathfrak{C}_l$, respectively, yield the coefficients $g_l$ as 
 \begin{align*}
 g_{l} &=\log(2) m_{1l} +  \log\left( 1 - \frac{\beta_T}{\xi}\right)m_{2l}+ \log(2) \left(\sum^{l-1}_{i = 3} m_{il}+ \sum^{N}_{i = l+1}m_{li}\right). 
 \end{align*}
 From (\ref{dij_first}) follows that
 \begin{align}
 g_l &= \log(2) m_{1l} +  \log\left( 1 - \frac{\beta_T}{\xi}\right)m_{2l}+ \frac{\log(2)}{\log\left(\frac{3}{4}\right) } \left(\sum^{l-1}_{i = 3} g_{il}+ \sum^{N}_{i = l+1}g_{li}\right).\label{dl_first}
 \end{align}
 The values of the coefficients $g_{ij}$ of the zero-one UQP (\ref{opt_prob_ooo}) have to be either $-2$ or $0$, which we obtain from (\ref{dij_first}) by the two steps below.
\begin{enumerate}
\item We choose that the transition $\mathfrak{I}_{ij}$ either occurs never or, independently of the nodes $i$ and $j$, $m_0$ times. Thus
\begin{align}\label{nIijnIl}
m_{ij} = m_0 r_{ij},
\end{align}
 where the binary variable $r_{ij}$ denotes whether the transition $\mathfrak{I}_{ij}$ occurs either never ($r_{ij} = 0$) or $m_0$ times ($r_{ij} = 1$). Then, the coefficients $g_{ij}$, given by (\ref{dij_first}), become 
 \begin{align*}
 g_{ij} = \log\left(\frac{3}{4}\right)  m_0 r_{ij}.
 \end{align*}
\item We multiply the objective function (\ref{fobj_UQP}) with a constant factor $\mu =  -2 /\log\left(\frac{3}{4}\right) > 0$ and divide by $m_0$, which yields the new objective function
\begin{align}
\tilde{f}_n(\hat{a}_{13}, ..., \hat{a}_{1N}) &= \frac{\mu}{m_0} f_n(\hat{a}_{13}, ..., \hat{a}_{1N}) \label{ftilde}\\
& = \sum^{N}_{i = 3} \sum^{N-1}_{j = i+1} c_{ij} \hat{a}_{1i} \hat{a}_{1j} + \sum^{N}_{l = 3} c_l \hat{a}_{1l} +c_\textrm{const},\nonumber
\end{align}
with the coefficients $c_{ij} = \mu  g_{ij}/m_0$, $c_l = \mu  g_l/m_0$ and $c_\textrm{const} = \mu  g_\textrm{const}/m_0$. The maximisation of $f_n(\hat{a}_{13}, ..., \hat{a}_{1N})$ is equivalent to the maximisation of $\tilde{f}_n(\hat{a}_{13}, ..., \hat{a}_{1N})$. As desired, the coefficients $c_{ij}$ attain the values $-2$ and $0$ for $r_{ij} = 1$ and $r_{ij}= 0$, respectively.
\end{enumerate}
 From (\ref{dl_first}), (\ref{nIijnIl}) and (\ref{ftilde}), we obtain the coefficients $c_l$ of the new objective function $\tilde{f}_n(\hat{a}_{13}, ..., \hat{a}_{1N})$ as 
\begin{align} 
c_l &= \frac{\mu}{m_0}  \log(2) m_{1l}   + \frac{\mu}{m_0}  m_{2l} \log\left( 1 - \frac{\beta_T}{\xi}\right)+ \frac{\mu}{m_0}  \frac{\log(2)}{\log\left(\frac{3}{4}\right) } \left(\sum^{l-1}_{i = 3} g_{il}+ \sum^{N}_{i = l+1}g_{li}\right) .\label{dlfinal}
\end{align}
Since $g_{il} = m_0 c_{il}/\mu$, equation (\ref{dlfinal}) is equivalent to
\begin{align} 
c_l &= \frac{m_{1l}}{m_0} \mu \log(2)    + \frac{ m_{2l} }{m_0}  \mu\log\left( 1 - \frac{\beta_T}{\xi}\right)+ \frac{\log(2)}{\log\left(\frac{3}{4}\right) } \left(\sum^{l-1}_{i = 3} c_{il}+ \sum^{N}_{i = l+1}c_{li}\right) .\label{skkksk}
\end{align}
By defining 
\begin{align}
\lambda_{+} &= \mu \log(2)   > 0 \label{lambdaPlus}\\
\lambda_{-} &= \mu \log\left( 1 - \frac{\beta_T}{\xi}\right) < 0 \label{lambdaMinus}\\
\eta_l &= \frac{\log(2)}{\log\left(\frac{3}{4}\right) } \left(\sum^{l-1}_{i = 3} c_{il}+ \sum^{N}_{i = l+1}c_{li}\right) \ge 0 \label{etaL},
\end{align}
it follows that (\ref{skkksk}) is equivalent to (\ref{clLEmm}). Hence, we have proved Lemma \ref{lemma:sis_as_UQP}.

\section{Proof of Lemma \ref{lemma:approaching}}
\label{appendix:approaching}

Equation (\ref{clLEmm}) shows that the coefficients $c_l$ are determined by the numbers $m_0$ and $m_{1l}$ of infectious transitions $\mathfrak{I}_{ij}$ and $\mathfrak{I}_{l}$ and by the number $m_{2l}$ of constant transitions $\mathfrak{C}_{l}$. The third addend $\eta_l$ in (\ref{clLEmm}) is constant with respect to the number $m_{1l}$, $m_{2l}$ and $m_0$ of occurrences of the viral state transitions $\mathfrak{I}_{ij}$, $\mathfrak{I}_{l}$ and $\mathfrak{C}_{l}$. We consider the two terms with which the coefficients $m_{1l}$ and $m_{2l}$ in equation (\ref{clLEmm}) are multiplied and denote them by $q_0 = \lambda_+ /m_0$ and $q_1 = \lambda_- / m_0$. It holds that $q_0 > 0$ and $q_1 <0$. Furthermore, if $m_0$ grows to infinity, then the absolute value of the two coefficients $q_0$ and $q_1$ becomes arbitrarily small. Thus, for a sufficiently large number $m_0$ of infectious transitions $\mathfrak{I}_{ij}$, we can choose the number $m_{1l}$ of infectious transitions $\mathfrak{I}_{l}$ and the number $m_{2l}$ of constant transitions $\mathfrak{C}_{l}$, such that the coefficient $c_l$, given by (\ref{clLEmm}), is arbitrarily close to any real number $b_l \in \mathbb{R}$. 

\section{Proof of Lemma \ref{lemma:sufficientlyClose}}
\label{appendix:sufficientlyClose}
We define the vector, which is composed of the optimisation variables of the zero-one UQP (\ref{01uqpsimpler}), as
\begin{align*}
y = (y_1, ..., y_N) \in \{0, 1\}^N.
\end{align*}
Furthermore, we denote the objective function of the zero-one UQP (\ref{01uqpsimpler}) by
\begin{align} \label{fobjUQPorig}
f_\textrm{obj}(y) = \sum^{N}_{i = 1} \sum^{N}_{j = i+1} b_{ij} y_i y_j + \sum^{N}_{l = 1} b_{l} y_l .
\end{align}
The coefficients $c_l$ given by (\ref{clLEmm}) do not precisely equal the coefficients $b_l$ for any finite numbers of transitions $m_0, m_{1l}, m_{2l}$. Instead, we have
\begin{align}\label{cltildedef}
c_l = b_l + \varepsilon_l,
\end{align}
with the error $\varepsilon_l$ on the $l$-th coefficient. The statement (\ref{epsiloncloseness}) implies that there is a finite number of transitions $m_0, m_{1l}, m_{2l}$, such that the error terms $\varepsilon_l$ are bounded by an arbitrarily small $\varepsilon_\textrm{max} \in \mathbb{R}^+$ and may be chosen to be non-negative:
\begin{align} \label{epsilonmax}
0 \le \varepsilon_l \le \varepsilon_\textrm{max}, \quad l = 1, ..., N.
\end{align}
Thus, when the coefficients $b_l$ in (\ref{fobjUQPorig}) are replaced by the distorted coefficients $c_l$ in (\ref{cltildedef}), the objective function $f_\textrm{obj}$, given by (\ref{fobjUQPorig}), is replaced by
\begin{align*}
\tilde{f}_\textrm{obj}(y)&= \sum^{N}_{i = 1} \sum^{N}_{j = i+1} b_{ij} y_i y_j + \sum^{N}_{l = 1} c_l y_l  \\
&= \sum^{N}_{i = 1} \sum^{N}_{j = i+1} b_{ij} y_i y_j + \sum^{N}_{l = 1} b_{l} y_l  + \sum^{N}_{l = 1} \varepsilon_l y_l. 
\end{align*}
More compactly, we obtain
\begin{align} \label{ftildecompact}
\tilde{f}_\textrm{obj}(y)&= f_\textrm{obj}(y)  + \varepsilon^T y, 
\end{align}
with the error vector $\varepsilon = (\varepsilon_1, ..., \varepsilon_N)^T$. 

Our aim is to show that the solution $\tilde{y}_\textrm{opt}$, or one of the solutions, to the zero-one UQP (\ref{01uqpsimpler}), with the objective function $\tilde{f}_\textrm{obj}$ given by (\ref{ftildecompact}), is also a solution to the original zero-one UQP (\ref{01uqpsimpler}) with the objective function $f_\textrm{obj}$ given by (\ref{fobjUQPorig}). Hence, the solution $\tilde{y}_\textrm{opt}$ would also be a solution to the maximum cut problem. More precisely, we want to show that
\begin{align}
\exists \tilde{y}_\textrm{opt} \in S_\textrm{opt}: \quad \tilde{f}_\textrm{obj}(\tilde{y}_\textrm{opt}) > \tilde{f}_\textrm{obj}(y) \quad \forall y \not \in S_\textrm{opt}, \label{conditionSolution}
\end{align}
where the set of solutions to the zero-one UQP (\ref{01uqpsimpler}) with the objective function $f_\textrm{obj}$, given by (\ref{fobjUQPorig}), is denoted as $S_\textrm{opt}$. We denote the value of the objective function $f_\textrm{obj}$, given by (\ref{fobjUQPorig}), evaluated at one of the elements in $S_\textrm{opt}$ as
\begin{align} \label{definitionfopt}
f_\textrm{opt} = f_\textrm{obj}(y), \quad y \in S_\textrm{opt}.
\end{align}
Furthermore, we define the gap from the optimal value $f_\textrm{opt}$ to the next largest value, that the objective function $f_\textrm{obj}$ attains, as
\begin{align}\label{fDelta}
 \begin{aligned}
\Delta f = &\underset{y}{\text{ min }} & &  f_\textrm{opt} -f_\textrm{obj}(y) & \\
 &\text{ s.t.} & &  y \not\in S_\textrm{opt}.&
\end{aligned} 
\end{align}
It holds $\Delta f \ge 1$, since the maximum cuts, given by the elements in $S_\textrm{opt}$, contain at least one more link than any suboptimal cut.

With the definitions above, we can show the statement (\ref{conditionSolution}) as follows. The equations (\ref{ftildecompact}) and (\ref{definitionfopt}) yield, for any $\tilde{y}_\textrm{opt} \in S_\textrm{opt}$ and any $y \not\in S_\textrm{opt}$, that
\begin{align*}
\tilde{f}_\textrm{obj}(\tilde{y}_\textrm{opt}) - \tilde{f}_\textrm{obj}(y) & =  f_\textrm{opt} - f_\textrm{obj}(y)  + \varepsilon^T (\tilde{y}_\textrm{opt} - y)\\
&\ge \Delta f  + \varepsilon^T (\tilde{y}_\textrm{opt} - y), 
\end{align*}
where the inequality follows from (\ref{fDelta}). Since the optimisation variables $y_l$ are either 0 or 1, we have $\tilde{y}_\textrm{opt} - y \ge -u$, where the inequality holds component-wise and $u = (1, ..., 1)^T \in \mathbb{R}^N$ denotes the all-one vector. As stated by (\ref{epsilonmax}), the error terms $\varepsilon_l$ are positive. Hence, we obtain
\begin{align}
\tilde{f}_\textrm{obj}(\tilde{y}_\textrm{opt}) - \tilde{f}_\textrm{obj}(y) & \ge \Delta f  - \varepsilon^T u\nonumber \\
& \ge \Delta f  - N \varepsilon_\textrm{max}, \label{jjjssial}
\end{align}
where the last inequality follows from (\ref{epsilonmax}). From the inequality (\ref{jjjssial}) we obtain
\begin{align*}
\tilde{f}_\textrm{obj}(\tilde{y}_\textrm{opt}) > \tilde{f}_\textrm{obj}(y) 
\end{align*}
if 
\begin{align*}
 \varepsilon_\textrm{max}  < \frac{1}{N} \le \frac{\Delta f }{N}.
\end{align*}

\section{Proof of Lemma \ref{lemma:redsizenetrecon}}
\label{appendix:redsizenetrecon}
 The objective function of the full-size SIS network reconstruction (\ref{maxlikelihood}) at time $n_2>n_1$ satisfies 
\begin{align}
f_{n_2}(\hat{A}) &=  \log\left(\mathrm{Pr}\left[x[1], ..., x[n_2]\big| \hat{A}\right]\right) \nonumber\\
&=\log\left(\mathrm{Pr}\left[x[1], ..., x[n_1]\big| \hat{A}\right]\right)  + \sum^{n_2-1}_{k = n_1} \log\left(\mathrm{Pr}\left[x[k+1]\big|x[k], \hat{A}\right]\right), \label{objectiveNEWWW}
\end{align}
which follows from the Markov property of the SIS process. We adjust the second addend of (\ref{objectiveNEWWW}) by constructing the viral state sequence $x[n_1+1], ..., x[n_2]$, such that the objective function $f_{n_2}$ at time $n_2$ attains the form (\ref{redSizeSISAmend}) in the second statement of Lemma \ref{lemma:redsizenetrecon}.

We divide the first statement of Lemma \ref{lemma:redsizenetrecon} into two parts: Firstly, we show in Subsection \ref{subsec:existence} how to construct a viral state sequence $x[n_1+1], ..., x[n_2]$, such that $(A_\textrm{ML})_{ij} = a_{ij}$ if $a_{ij}=1$. Secondly, we show in Subsection \ref{subsec:absence} how to construct a viral state sequence $x[n_1+1], ..., x[n_2]$, such that $(A_\textrm{ML})_{ij} = a_{ij}$ if $a_{ij}=0$. The second statement of Lemma \ref{lemma:redsizenetrecon} is proved in Subsection \ref{subsec:firstStatementLem}.

\subsection{Enforce Existence of Links}
\label{subsec:existence}
We denote the set of links $(i, j) \in \mathcal{L}$ in the first statement of Lemma \ref{lemma:redsizenetrecon} by 
\begin{align*}
\bar{\mathcal{L}} = \left\{(i, j) \in \mathcal{L}\big| (i, j)= (1, 2) \lor (i \ge 2\land j\ge 2) \right\}.
\end{align*}
We aim to construct a viral state sequence such that the ML estimate (\ref{maxlikelihood}) satisfies $(A_\textrm{ML})_{ij} = a_{ij}$ if the element of the true adjacency matrix is $a_{ij}=1$ for all links $(i, j) \in \bar{\mathcal{L}}$. We make use of the following transition: If a node $j$ gets infected at time $k+1$ and only node $i$ has been infected at time $k$, then there must be a link between node $i$ and $j$. We define the infectious transition, followed by a curing of node $i$, as
\begin{align}\label{viraltransE}
\mathfrak{E}_{ij}= \left\{x[k+2] = e_j, x[k+1] = e_i +e_j  \big| x[k] = e_i \right\}.
\end{align}
The probability of the transition $\mathfrak{E}_{ij}$ follows from (\ref{infectionTrans}) and (\ref{healTransition}) as
\begin{align}\label{prijoooo}
\mathrm{Pr}\left[\mathfrak{E}_{ij}\big| \hat{A}\right] = \begin{cases}
\beta_T \delta_T \quad &\text{if} \quad \hat{a}_{ij} = 1, \\
0 &\text{if} \quad \hat{a}_{ij} = 0 .
\end{cases}
\end{align}
We construct the viral state sequence $x[n_1+1], ..., x[n_2]$ such that it contains $\mathfrak{E}_{ij}$ at least once for all links $(i, j) \in \bar{\mathcal{L}}$. Then, it follows from (\ref{prijoooo}) that if the underlying matrix has the element $a_{ij}=1$ but the solution candidate $\hat{A}$ contains a zero element $\hat{a}_{ij}=0$ for any link $(i, j) \in \bar{\mathcal{L}}$, then the objective function of (\ref{maxlikelihood}) becomes zero: $\mathrm{Pr}\left[ x[1], ..., x[n_2] \big| \hat{A}\right] = 0$. Thus, the solution $A_\textrm{ML}$ to the full-size SIS network reconstruction (\ref{maxlikelihood}) with the objective function (\ref{objectiveNEWWW}) has to satisfy
\begin{align} \label{existenceLinks} 
 (A_\textrm{ML})_{ij} = 1 \quad \text{if} \quad a_{ij} = 1, \quad \forall(i, j) \in \bar{\mathcal{L}}
\end{align}

\subsection{Enforce Absence of Links}
\label{subsec:absence}
We aim to construct a viral state sequence such that the ML estimate (\ref{maxlikelihood}) satisfies $(A_\textrm{ML})_{ij} = a_{ij}$ if the element of the true adjacency matrix is $a_{ij}=0$. We observe the following: If solely a node $l$ is infected at time $k$ and the viral state $x[k]$ does not change from time $k$ to $k+1$, then the existence of a link from node $l$ to another node $m$ becomes \textit{less probable}, which follows from (\ref{constTrans}). For a node $l\ge 2$, we define the constant viral state transition 
\begin{align}\label{viraltransA}
\mathfrak{A}_l= \left\{x[k+1] = e_l  \big| x[k] = e_l \right\}.
\end{align}
The probability of the transition above follows from (\ref{constTrans}) as 
\begin{align}
\mathrm{Pr}\left[\mathfrak{A}_l\big| \hat{A}\right] &= 1 - \delta_T - \sum^N_{m=1} \beta_T \hat{a}_{ml}. \label{probabililliil}
\end{align}
We consider that the transition $\mathfrak{A}_l$ successively occurs $\kappa_l$ times from some time $k_0 \in \{n_1 +1, ..., n_2\}$ to time $k_0 + \kappa_l$. For ease of exposition and without loss of generality, we assume that $k_0= n_1 + 1$. Hence, the transition $\mathfrak{A}_l$ multiply occurs from time $n_1+1$ to time $n_1 +\kappa_l +1$. Then, the probability of the transition sequence from time $n_1+1$ to $n_1 +\kappa_l +1$ follows from (\ref{probabililliil}) as 
\begin{align}
 \log\left(\mathrm{Pr}\left[x[n_1 + \kappa_l+1] = x[n_1 + \kappa_l]= ...= x[n_1 + 2] = e_l \big|x[n_1+1] = e_l, \hat{A}\right]\right) = \kappa_l \log\left( \mathrm{Pr}\left[\mathfrak{A}_l\big| \hat{A}\right]\right).\label{kkkkkssasd}
\end{align}
The objective function of the full-size SIS network reconstruction (\ref{maxlikelihood}) at time $n_1+\kappa_l+1$ becomes
\begin{align*}
f_{n_1+\kappa_l+1} (\hat{A}) &= f_{n_1}(\hat{A}) + \kappa_l \log\left( \mathrm{Pr}\left[\mathfrak{A}_l\big| \hat{A}\right]\right)  \\
&=f_{n_1}(\hat{A}) + \kappa_l \log\left(1 - \delta_T - \beta_T  \sum^N_{m=1} \hat{a}_{ml}  \right),
\end{align*}
where the last equality follows from (\ref{probabililliil}) and (\ref{kkkkkssasd}). By defining the degree of node $l$ minus the element $\hat{a}_{1l}$ as
\begin{align}\label{kkqquqqq}
d_l(\hat{A}) = \sum^N_{m=2} \hat{a}_{ml},
\end{align}
we finally formulate the objective function of the SIS network reconstruction (\ref{maxlikelihood}) at time $n_1+\kappa_l+1$ as
\begin{align}\label{bnewwwkjbhbba}
f_{n_1+\kappa_l+1} (\hat{A}) &= f_{n_1}(\hat{A})  + \kappa_l \log\left(1 - \delta_T - \beta_T d_l(\hat{A})   -\beta_T \hat{a}_{1l} \right).
\end{align}
Based on the above formulation of the objective function (\ref{bnewwwkjbhbba}), we will show that if the number $\kappa_l$ of occurrences of the transition $\mathfrak{A}_l$ is great enough, then the solution $A_\textrm{ML}$ to the SIS network reconstruction (\ref{maxlikelihood}) satisfies $(A_\textrm{ML})_{ml} = a_{ml}$ for all nodes $m \ge 2$.

Due to $a_{ij} = 1 \Rightarrow \hat{a}_{ij} = 1$ for $i,j\ge 2$ as stated by (\ref{existenceLinks}), the ML estimate $A_\textrm{ML}$ has at least as many links between the nodes $i,j\ge 2$ as the true adjacency matrix $A$. Thus, the degree $d_l(A)$ of node $l$ of the true adjacency matrix $A$, given by (\ref{kkqquqqq}) when replacing $\hat{a}_{ml}$ by $a_{ml}$, is upper bounded by 
\begin{align}
d_l(A_\textrm{ML}) \ge d_l(A).  \label{kkkssaaa}
\end{align}
Furthermore, since $(A_\textrm{ML})_{ij} = a_{ij}$ for $i,j \ge 2$, we obtain 
\begin{align}
d_l(A_\textrm{ML}) = d_l(A) \Leftrightarrow (A_\textrm{ML})_{ml} = a_{ml} \quad \forall m=2, ..., N. \label{llldl}
\end{align}
Hence, it is sufficient to show that the ML estimate $A_\textrm{ML}$ satisfies $d_l(A_\textrm{ML}) = d_l(A)$ in order to prove the second statement of Lemma \ref{lemma:redsizenetrecon}.

In the following, we consider two solution candidates to the full-size SIS network reconstruction (\ref{maxlikelihood}): two matrices $\hat{A}_1$ and $\hat{A}_2$. We assume that the first row (and column) of the two solution candidates are equal, i.e. 
\begin{align}\label{equalA1A2}
(\hat{A}_1)_{1m}  = (\hat{A}_2)_{1m},\quad m = 1, ..., N.
\end{align}
From (\ref{equalA1A2}) follows that the matrices $\hat{A}_1$ and $\hat{A}_2$ result in the same objective value for the \textit{reduced-size} SIS network reconstruction (\ref{maxlikelireduced}), since the optimisation is only with respect to the matrix elements $\hat{a}_{1m}$ for $m = 3, ..., N$.  We consider that the two matrices $\hat{A}_1$ and $\hat{A}_2$ differ as follows. On the one hand, the first solution candidate $\hat{A}_1$ is a matrix that satisfies (\ref{kkkssaaa}) with equality:
\begin{align}\label{kkkssaa}
d_l(\hat{A}_1)= d_l(A).
\end{align}
On the other hand, the second solution candidate $\hat{A}_2$ is a matrix that does not satisfy (\ref{kkkssaaa}) with equality:
\begin{align*}
d_l(\hat{A}_2) > d_l(A).
\end{align*}
To check which of the matrices $\hat{A}_1$ and $\hat{A}_2$ yields a greater objective value of the \textit{full-size} SIS network reconstruction (\ref{maxlikelihood}) at time $n_1 +\kappa_l +1$, we compute the difference of the objective function (\ref{bnewwwkjbhbba}) as
\begin{align}
f_{n_1+\kappa_l+1} (\hat{A}_1) - f_{n_1+\kappa_l+1} (\hat{A}_2) &=  f_{n_1} (\hat{A}_1)-f_{n_1} (\hat{A}_2)+ \kappa_l \log\left(1 - \delta_T - \beta_T  d_l(\hat{A}_1)-\beta_T (\hat{A}_1)_{1l} \right) \nonumber\\
&\quad \quad-\kappa_l \log\left(1 - \delta_T - \beta_T  d_l(\hat{A}_2)-\beta_T (\hat{A}_2)_{1l} \right) \nonumber\\
 &=  f_{n_1} (\hat{A}_1)-f_{n_1} (\hat{A}_2)+ \kappa_l \gamma_l, \label{lldddkksa}
\end{align}
where
\begin{align*}
\gamma_l =  \log\left(\frac{1 - \delta_T - \beta_T  d_l(\hat{A}_1)-\beta_T (\hat{A}_1)_{1l}}{1 - \delta_T - \beta_T  d_l(\hat{A}_2)-\beta_T (\hat{A}_2)_{1l}} \right) .
\end{align*}
It holds $d_l(\hat{A}_2)> d_l(A) =  d_l(\hat{A}_1)$ and, as stated by (\ref{equalA1A2}), $(\hat{A}_1)_{1l} = (\hat{A}_2)_{2l}$. Thus, it holds $\gamma_l >0$. Since the difference $f_{n_1} (\hat{A}_1)-f_{n_1} (\hat{A}_2)$ is finite, there is a number $\kappa_l \in \mathbb{N}$ of occurrences of the transition $\mathfrak{A}_l$, such that the right-hand side of (\ref{lldddkksa}) is positive, which implies $f_{n_1+\kappa_l+1} (\hat{A}_1) > f_{n_1+\kappa_l+1} (\hat{A}_2)$. Hence, the matrix $\hat{A}_1$ results in a greater objective value of the optimisation problem (\ref{maxlikelihood}) than the matrix $\hat{A}_2$ for a sufficiently large number of transitions $\kappa_l$, and the matrix $\hat{A}_2$ cannot be a solution of (\ref{maxlikelihood}). Thus, if the number of transitions $\kappa_l$ is sufficiently great, then the matrix $A_\textrm{ML}$ that solves the full-size SIS network reconstruction (\ref{maxlikelihood}) has to be of the kind $\hat{A}_1$ and satisfy equation (\ref{kkkssaa}): $d_l(A_\textrm{ML}) = d_l(A)$. As stated by (\ref{llldl}), the equation $d_l(A_\textrm{ML}) = d_l(A)$ is equivalent to $(A_\textrm{ML})_{ml} = a_{ml}$ for all nodes $m \ge 2$.

In order to complete the proof of the first statement of Lemma \ref{lemma:redsizenetrecon}, it needs to hold $(A_\textrm{ML})_{ml} = a_{ml}$ for all nodes $m \ge 2$ and \textit{additionally} for all nodes $l \ge 2$. We achieve $(A_\textrm{ML})_{ml} = a_{ml}$ for all nodes $m,l \ge 2$ as follows. We design the viral state sequence $x[n_1+1], ..., x[n_2]$ such that it solely consists of two kind of viral transitions: Firstly, the transitions $\mathfrak{E}_{ij}$ given by (\ref{viraltransE}) for all links $(i, j) \in \bar{\mathcal{L}}$. Secondly, the transitions $\mathfrak{A}_l$ given by (\ref{viraltransA}), which occur $\kappa_l$ times for all nodes $l \ge 2$. 

Finding a \emph{shortest}\footnote{There is no necessity to use the shortest walk here, as long as the walk visits every link $(i, j) \in \bar{\mathcal{L}}$.} walk which traverses every link in a graph is known as the Chinese Postman Problem (CPP) or route intersection problem \cite{edmonds1973matching}. The CPP is solvable in polynomial time. Since every link $(i, j) \in \bar{\mathcal{L}}$ has to be traversed by an infection, we define the graph $\bar{G} = (\mathcal{N}, \bar{\mathcal{L}})$ and denote the solution to the CPP as
\begin{align*}
\left((i_1, j_1), ..., (i_r, j_r) \right)= \text{CPP}\left(\bar{G} \right),
\end{align*}
where $i_1,..., i_r$ and $j_1, ..., j_r$ denote the successive nodes of the walk, where $i_{l+1} = j_{l}$, and $(i_l, j_l)\in \bar{\mathcal{L}}$ denote the traversed links. Algorithm \ref{alg:exploration_walk_two} illustrates in pseudo-code how the required viral state sequence $x[n_1+1], ..., x[n_2]$ can be constructed.

\begin{algorithm}
\caption{Construction of Viral State Sequence for Full-Size SIS Network Reconstruction}
\begin{algorithmic}[1]
\State \textbf{Input: } graph $\bar{G} = (\mathcal{N}, \bar{\mathcal{L}})$, multiplicities $\kappa_2, ..., \kappa_N$
\State \textbf{Output: } viral state sequence $x[n_1+1], ..., x[n_2]$
\State $\left((i_1, j_1), ..., (i_r, j_r) \right) \gets \text{CPP}\left(\bar{G} \right)$
\State $\mathcal{D} \gets \emptyset$ \Comment{set of visited nodes}
\State $k \gets 1$
\For{$l= 1, ..., r$}
\State $p \gets i_l$, $q \gets j_l$
\State $x[k] \gets e_p$
\State $k \gets k+1$
\If{$p \not\in \mathcal{D} \land p \neq 1$}
\State $(x[k], x[k+1], ..., x[k+\kappa_p-1])\gets (e_p, e_p,..., e_p)$ \Comment{$\kappa_p$ times constant transition $\mathfrak{A}_p$}
\State $k \gets k+\kappa_p$
\State $\mathcal{D} \gets \mathcal{D} \cup \{p\}$
\EndIf
\State $x[k] \gets e_p+e_q$ \Comment{essential part of transition $\mathfrak{E}_{pq}$ (infection from node $p$ to $q$)}
\State $k \gets k+1$
\EndFor
\State $x[k] \gets e_q$
\State $k \gets k+1$
\end{algorithmic}
\label{alg:exploration_walk_two}
\end{algorithm}

With the construction of the viral state sequence $x[n_1+1], ..., x[n_2]$ as described by Algorithm \ref{alg:exploration_walk_two}, the objective function (\ref{objectiveNEWWW}) at time $n_2$ becomes
\begin{align} \label{fn2NEWW}
f_{n_2}(\hat{A}) = f_{n_1}(\hat{A})+ \sum^N_{l=2}\kappa_l \log\left(1 - \delta_T - \beta_T  d_l(\hat{A})-\beta_T \hat{a}_{1l} \right) + \zeta,
\end{align}
where $\zeta$ is finite and depends on the transition probabilities $\mathrm{Pr}\left[\mathfrak{E}_{ij}\big| \hat{A}\right] = \beta_T \delta_T \hat{a}_{ij}$, given by (\ref{prijoooo}), for the links $(i, j) \in \bar{\mathcal{L}}$. 

 By choosing the number of transitions $\kappa_l$ sufficiently great for all nodes $l \ge 2$, we finally obtain that the matrix $A_\textrm{ML}$ that solves the full-size SIS network reconstruction (\ref{maxlikelihood}) at time $n_2$ with the objective function (\ref{fn2NEWW}) has to satisfy $(A_\textrm{ML})_{ij} = a_{ij}$ for all links $(i, j) \in \bar{\mathcal{L}}$.
 
\subsection{Second Statement of Lemma \ref{lemma:redsizenetrecon}}
\label{subsec:firstStatementLem}
 As given by the first statement of Lemma \ref{lemma:redsizenetrecon}, the solution $\hat{A}_\textrm{ML}$ to the full-size SIS network reconstruction (\ref{maxlikelihood}) with the objective function (\ref{fn2NEWW}) has to satisfy $(\hat{A}_\textrm{ML})_{ij} = a_{ij}$ for $(i, j) \in \bar{\mathcal{L}}$. Hence, the full-size SIS network reconstruction problem (\ref{maxlikelihood}) at time $n_2$ becomes
 \begin{align} \label{sisFullsizeAlt}
 \begin{aligned}
\hat{A}_\textrm{ML} = &\underset{\hat{A}}{\textup{ arg max }} & &  f_{n_1}(\hat{A})+ \sum^N_{l=2}\kappa_l \log\left(1 - \delta_T - \beta_T  d_l(\hat{A})-\beta_T \hat{a}_{1l} \right) + \zeta & \\
 &\textup{ s.t. } & &  \hat{a}_{ij} \in \{0, 1\}, \quad i,j= 1, ..., N,&\\
 &&& \hat{a}_{ij} = \hat{a}_{ji}, \quad\quad i,j= 1, ..., N,&\\
 &&& \hat{a}_{ii} = 0, \quad\quad\quad i= 1, ..., N,&\\
 &&& \hat{a}_{ij} = a_{ij}, \quad\quad \forall(i, j) \in \bar{\mathcal{L}}, &
\end{aligned}
\end{align}
where the objective function follows from (\ref{objectiveNEWWW}) and (\ref{fn2NEWW}). Since the optimisation variables $\hat{a}_{ij}$ in (\ref{sisFullsizeAlt}) are fixed to $a_{ij}$ for $(i, j) \in \bar{\mathcal{L}}$, the optimisation takes place only with respect to the elements $\hat{a}_{13}, ..., \hat{a}_{1N}$. Furthermore, the term $\zeta$ does not depend on the links $\hat{a}_{13}, ..., \hat{a}_{1N}$ and can be omitted in (\ref{sisFullsizeAlt}). By the formal replacement
 \begin{align*}
\mathrm{Pr} \left[x[1], ..., x[n_1] \Big| \hat{a}_{13}, ..., \hat{a}_{1N} \right] = f_{n_1}(\hat{A}) \quad\text{if} \hat{a}_{ij} = a_{ij} \quad \forall (i, j) \in \bar{\mathcal{L}},
\end{align*}
we obtain the second statement of Lemma \ref{lemma:redsizenetrecon}.

\section{Proof of Theorem \ref{theorem:NPhard}}
\label{appendix:NPhard}

To show that the optimisation problem (\ref{redSizeSISAmend}) is NP-hard, we consider the addends in the sum of its objective function, which equal
\begin{align} \label{jkjjjjjjjjjjjjjjjjjjjjjjjjjjjj}
\kappa_l \log\left(1 - \delta_T - \beta_T  d_l(\hat{A})-\beta_T \hat{a}_{1l} \right)  &=\kappa_l \log\left(1 - \delta_T - \beta_T  d_l(A) \right) +\hat{a}_{1l}  \kappa_l\log\left(\frac{1 - \delta_T - \beta_T  d_l(A)-\beta_T}{1 - \delta_T - \beta_T d_l(A)}\right), 
\end{align}
where we used the fact that $d_l(A) = d_l(\hat{A})$ as stated Subsection \ref{subsec:absence}. The first addend in (\ref{jkjjjjjjjjjjjjjjjjjjjjjjjjjjjj}) is constant with respect to the links $\hat{a}_{1m}$ for all nodes $m$ and thus the term has not to be considered in the optimisation problem (\ref{redSizeSISAmend}). However, the second addend in (\ref{jkjjjjjjjjjjjjjjjjjjjjjjjjjjjj}) given by $\hat{a}_{1l} \chi_l$, where
\begin{align*}
\chi_l =  \kappa_l\log\left(\frac{1 - \delta_T - \beta_T  d_l(A)-\beta_T}{1 - \delta_T - \beta_T d_l(A)}\right),
\end{align*}
 is not constant with respect to the elements $\hat{a}_{1m}$ and has to be considered in the optimisation problem (\ref{redSizeSISAmend}). Hence, the optimisation problem (\ref{redSizeSISAmend}) is of the form (\ref{opt_prob_ooo}) when the coefficients $c_l$ in (\ref{opt_prob_ooo}) are replaced by $c_l + \chi_l$, and the optimisation problem (\ref{redSizeSISAmend}) becomes
\begin{align}\label{jjjdddaaa}
 \begin{aligned}
&\underset{\hat{a}_{13}, ..., \hat{a}_{1N}}{\textup{max }} & & \sum^{N}_{i = 3} \sum^{N}_{j = i+1} c_{ij} \hat{a}_{1i} \hat{a}_{1j} + \sum^{N}_{l = 3} (c_l + \chi_l) \hat{a}_{1l} & \\
 &\textup{s.t. } & &  \hat{a}_{1i} \in \{0, 1\}, \quad i= 3, ..., N.&
\end{aligned}
\end{align}
Since the term $\chi_l$ is constant with respect to the elements $\hat{a}_{1m}$, it follows from Lemma \ref{lemma:approaching} that the coefficient $c_l$ can be set such that the coefficient $(c_l + \chi_l)$ approaches any real number arbitrarily close. If the coefficients $c_{ij}$ in (\ref{jjjdddaaa}) equal the coefficients $b_{ij}$ in the zero-one UQP (\ref{01uqpsimpler}) and the difference of the coefficients $(c_l + \chi_l)$ in (\ref{jjjdddaaa}) to the coefficients $b_{l}$ in (\ref{01uqpsimpler}) is positive and smaller than $1/N$, then it follows from Lemma \ref{lemma:sufficientlyClose} that the solution to (\ref{jjjdddaaa}) is also a solution to (\ref{01uqpsimpler}). Hence, solving the optimisation problem (\ref{jjjdddaaa}), which resulted from the full-size SIS network reconstruction (\ref{maxlikelihood}) as stated by Lemma \ref{lemma:redsizenetrecon}, implies solving the NP-hard zero-one UQP (\ref{01uqpsimpler}). 
 

\begin{thebibliography}{10}

\bibitem{prasseSIS}
B.~Prasse and P.~Van~Mieghem, ``Exact network reconstruction from complete
  {SIS} nodal state infection information seems infeasible,'' {\em Submitted}.

\bibitem{bodlaender1991complexity}
H.~L. Bodlaender, {\em On the complexity of the maximum cut problem}, vol.~91.
\newblock Unknown Publisher, 1991.

\bibitem{van2014performance}
P.~Van~Mieghem, {\em Performance {A}nalysis of {C}omplex {N}etworks and
  {S}ystems}.
\newblock Cambridge University Press, 2014.

\bibitem{karelCutSizeBounds}
K.~Devriendt and P.~Van~Mieghem, ``Tighter spectral bounds for the cut size,
  based on laplacian eigenvectors,'' {\em Submitted}.

\bibitem{van2015epidemic}
P.~Van~Mieghem and K.~Devriendt, ``An epidemic perspective on the cut size in
  networks,'' {\em Delft University of Technology}, vol.~1, no.~19, 2015.

\bibitem{garey1976some}
M.~R. Garey, D.~S. Johnson, and L.~Stockmeyer, ``Some simplified {NP}-complete
  graph problems,'' {\em Theoretical computer science}, vol.~1, no.~3,
  pp.~237--267, 1976.

\bibitem{cormen2009introduction}
T.~H. Cormen, {\em Introduction to algorithms}.
\newblock MIT press, 2009.

\bibitem{caprara2008constrained}
A.~Caprara, ``Constrained 0--1 quadratic programming: Basic approaches and
  extensions,'' {\em European Journal of Operational Research}, vol.~187,
  no.~3, pp.~1494--1503, 2008.

\bibitem{boros2002pseudo}
E.~Boros and P.~L. Hammer, ``Pseudo-boolean optimization,'' {\em Discrete
  Applied Mathematics}, vol.~123, no.~1-3, pp.~155--225, 2002.

\bibitem{rosenberg1975reduction}
I.~G. Rosenberg, ``Reduction of bivalent maximization to the quadratic case,''
  {\em Cahiers du Centre d’etudes de Recherche Operationnelle}, vol.~17,
  pp.~71--74, 1975.

\bibitem{garey1979guide}
M.~R. Garey and D.~S. Johnson, ``A {G}uide to the {T}heory of
  {NP}-{C}ompleteness,'' {\em WH Freemann, New York}, vol.~70, 1979.

\bibitem{picard1975minimum}
J.-C. Picard and H.~D. Ratliff, ``Minimum cuts and related problems,'' {\em
  Networks}, vol.~5, no.~4, pp.~357--370, 1975.

\bibitem{pardalos1991graph}
P.~M. Pardalos and S.~Jha, ``Graph separation techniques for quadratic zero-one
  programming,'' {\em Computers \& Mathematics with Applications}, vol.~21,
  no.~6-7, pp.~107--113, 1991.

\bibitem{barahona1986solvable}
F.~Barahona, ``A solvable case of quadratic 0--1 programming,'' {\em Discrete
  Applied Mathematics}, vol.~13, no.~1, pp.~23--26, 1986.

\bibitem{gartner2012approximation}
B.~G{\"a}rtner and J.~Matousek, {\em Approximation algorithms and semidefinite
  programming}.
\newblock Springer Science \& Business Media, 2012.

\bibitem{edmonds1973matching}
J.~Edmonds and E.~L. Johnson, ``Matching, euler tours and the chinese
  postman,'' {\em Mathematical programming}, vol.~5, no.~1, pp.~88--124, 1973.

\end{thebibliography}
\end{document}